\documentclass[letterpaper, 10 pt, conference]{ieeeconf}      
\IEEEoverridecommandlockouts                              
\usepackage{amsthm}
\usepackage{cite}\usepackage{hyperref}
\usepackage{algorithmic}
\usepackage{graphicx}
\usepackage{textcomp}
\usepackage{amsmath,amssymb,bm,bbm,mathrsfs,amscd}
\usepackage{calc}
\usepackage{color}
\usepackage{dsfont}
\usepackage{graphicx}
\usepackage{epstopdf}
\usepackage{epsfig}
\usepackage{tikz}
\usepackage{amsfonts}
\usepackage{rotating}
\usepackage{mathtools}
\usepackage{color}
\usepackage{subfig}
\usepackage{enumerate,pgfplots}
\usepackage[font=scriptsize]{caption}
\newtheorem{theorem}{Theorem}

\newtheorem{proposition}{Proposition}
\newtheorem{lemma}{Lemma}

\newtheorem{remark}{Remark}
\newtheorem{assumption}{Assumption}
\newcommand{\ba}{\begin{array}}
\newcommand{\ea}{\end{array}}

\newcommand{\be}{\begin{equation}}
\newcommand{\ee}{\end{equation}}

\newcommand{\ds}{\displaystyle}

\newcommand{\mc}{\mathcal}

\def\1{\boldsymbol{1}}

\def\diag{{\rm diag}\,}
\tikzstyle{v_c}=[circle, draw,inner sep=2pt, minimum width=12pt, color=blue]
\tikzstyle{v_a}=[circle, draw,inner sep=2pt, minimum width=12pt, color=red]
\tikzstyle{edge} = [draw,thick,-,font=\small ]
\tikzstyle{label} = [draw,fill=black,font=\normalsize]

\def\BibTeX{{\rm B\kern-.05em{\sc i\kern-.025em b}\kern-.08em
    
T\kern-.1667em\lower.7ex\hbox{E}\kern-.125emX}}
	\title{\LARGE \bf On Leadership Emergence in Opinion Dynamics on Social Networks
}
	\author{Martina~Alutto, Lorenzo Zino, Karl H. Johansson, and Angela Fontan
\thanks{M.~Alutto, K. H. Johansson and A. Fontan are with the Department of Decision and Control Systems, School of Electrical Engineering and Computer Science, KTH Royal Institute of Technology, Stockholm, Sweden (\{alutto; angfon; kallej\}@kth.se). They are also affiliated with the research center Digital~Futures,~Stockholm,~Sweden.
L. Zino is with the Department of Electronics and Telecommunications, Politecnico di Torino, Italy (lorenzo.zino@polito.it). This work was supported by the Wallenberg AI, Autonomous Systems and Software Program (WASP) funded by the Knut and Alice Wallenberg Foundation.}
}

\begin{document}

\maketitle
\thispagestyle{empty}
\pagestyle{empty}

\begin{abstract}
	Leadership in social groups emerges dynamically through interaction and opinion exchange. Empirical evidence indicates that individuals expressing strong opinions tend to gain influence, while sustained leadership critically depends on maintaining alignment with the surrounding social context. Motivated by these observations, we introduce a coupled dynamical model describing the simultaneous evolution of opinions and leadership in a networked population. Extending the Friedkin-Johnsen framework, we represent leadership as a time-varying susceptibility to social influence, which evolves according to a game-theoretic mechanism, consistent with social psychology evidence. Within this setting, agents strengthen their leadership by expressing decisive yet socially coherent opinions, whereas misalignment with the collective state results in a loss of influence. We analyze the coupled dynamics and establish sufficient conditions to identify which agents necessarily emerge as leaders and which act as followers in the social network.
\end{abstract}

\section{Introduction}

Mathematical modeling of complex social systems is an area of growing interest thanks to its ability to establish tools to increase the understanding of our societies, predict their emergent behavior, and design intervention policies to guide them to desired collective behavior. The systems and control community has developed analytical tools to study complex social phenomena, particularly through the lens of opinion dynamics~\cite{friedkin2015_socialsurvey,proskurnikov2017tutorial,anderson2019IJAC,Mei2022,Bernardo2024}, helping unveiling the mechanisms of critical societal phenomena such as the emergence of polarization and social clustering~\cite{Leonard2021,DePasquale2022}, pluralistic ignorance~\cite{Hassan2023tac}, persistent fluctuations~\cite{Acemoglu2013}, and decision-making processes in the presence of antagonistic interactions~\cite{fontan2021role}.

A key feature of social systems is leadership, which is the ability to influence others in social groups. Leadership has a key role in shaping the collective behavior of the group. Several theoretical and empirical studies from the social psychology literature suggest that leadership is not a static feature~\cite{Badura2022,klein2023respondents}. In fact, leaders in social groups dynamically emerge through interactions and opinion exchange~\cite{Nakayama2019,klein2023respondents}. On the one hand, people with strong and extreme beliefs are typically more confident in sharing their opinions and more committed. Hence, they tend to become more prominent in group discussions and are less susceptible to social influence~\cite{VanSwol01072009}, as confirmed by empirical studies~\cite{brandt2015unthinking}. On the other hand, sustained leadership critically depends on maintaining alignment with the surrounding social context~\cite{isenberg1986group,Surowiecki2004,hillman2023social}. In summary, leadership shapes the dynamics of social networks and dynamically emerges from nontrivial mechanisms, driven by social interactions and opinion sharing.

The spontaneous emergence of leadership in social networks is underexplored. This is mostly because classical models typically assume that leadership and susceptibility to others' opinions are constant~\cite{friedkin2015_socialsurvey,proskurnikov2017tutorial}. Some efforts have recently been made to account for dynamically changing social influence by means of appraisal dynamics~\cite{jia2015opinion,ohlin2022achieving,tian2021social,wang2024social,Kang2022,Liu2024}. However, in these works, the emergence of leadership is modeled by assuming exogenous (and typically constant) differences in stubbornness or authority, but with no impact of the social interactions. Recently, in~\cite{Alutto2026ifac}, a co-evolutionary model of leadership emergence and opinion formation was proposed, and some preliminary results have been established under time-scale separation assumptions. However, the complexity of that model limits the possibility of using it to investigate the concurrency of leadership emergence and opinion formation, calling for the development of new modeling approaches.

In this paper, we fill this gap by proposing and analyzing a novel framework that captures leadership emergence in opinion dynamics on social networks. Specifically, grounded in the social psychology literature on leadership emergence~\cite{VanSwol01072009,brandt2015unthinking,Surowiecki2004,isenberg1986group}, we propose a game-theoretic formulation of leadership in social systems and we use a replicator equation to capture its dynamic evolution~\cite{taylor1978replicator,Sandholm2010}. Then, we couple this equation with a classical opinion dynamics model, namely, Taylor's model \cite{taylor1968towards}, which is the continuous-time counterpart of the well-known Friedkin--Johnsen model~\cite{Friedkin1990}. This coupling yields a system of nonlinear ordinary differential equations (ODEs), in which the susceptibility to peer influence in the opinion dynamics is affected by the agent's leadership, while the latter is impacted by the strength of the agent's opinion and the willingness to align with their social circle, with a tradeoff that is parsimoniously regulated by a parameter capturing the agent's behavioral and cognitive characteristics. 

After formalizing the coupled model of leadership emergence and opinion formation, we establish general well-posedness results and we provide a characterization of equilibrium points. Then, after illustrating the richness of the system's behavior via numerical simulations, we focus our analysis on a key problem in social networks: determining who will emerge as a leader within a social group. 
In particular, by combining positively invariant set arguments with a careful analysis of the sign of the leadership dynamics, we derive sufficient conditions for some agents to asymptotically emerge as leaders. These conditions depend on the agents' characteristics, their initial state, and the initial opinions of their neighbors. Once a set of leaders have emerged, we also establish sufficient conditions to identify those who will necessarily act as followers. 

The remainder of the paper is organized as follows. In Section~\ref{sec:model}, we present the mathematical model. In Section~\ref{sec:general}, we present some general results on the system and illustrate its emergent behavior through some simulations. In Section~\ref{sec:results}, we present our main theoretical results. Section~\ref{sec:conclusion} concludes the paper and outlines future research directions. 

\section{Leadership-opinion Model}\label{sec:model}
\subsection{State variables}
We consider a social network composed of $n$ agents. The interaction structure is described by a matrix $W \in \mathbb{R}^{n \times n}$, which is assumed to be row-stochastic, i.e., $\sum_{j=1}^n W_{ij} = 1$ for all $i=1, \dots, n$. The entry $W_{ij}$ represents the weight that agent $i$ assigns to the opinion of agent $j$. 
Each agent $i \in \{1,\dots,n\}$ is characterized by two state variables:
\begin{itemize}
	\item An opinion variable $x_i(t) \in [-1,1]$, representing the belief of agent $i$ on a given topic;
	\item A leadership variable $y_i(t) \in [0,1]$, representing the degree of leadership of agent $i$, whereby in the limit cases, $y_i(t)=0$ means that agent $i$'s opinion is solely influenced by the opinions shared by others, while $y_i(t)=1$ means that $i$ acts as a leader and is not influenced by others.
\end{itemize}

Here, consistent with the literature~\cite{Surowiecki2004,isenberg1986group,VanSwol01072009,brandt2015unthinking,klein2023respondents,hillman2023social}, we assume that both the opinion and the leadership evolve over time, and we denote by $z_i(t)=(x_i(t),y_i(t))$ the variable that characterizes the state of agent $i$ at time $t\geq0$. The states of all agents are gathered in the matrix $z(t)=(x(t),y(t))\in [-1,1]^n\times[0,1]^n$.

\subsection{Game-theoretic model of leadership}
We formulate the leadership dynamics according to a game-theoretic mechanism. In particular, for a generic agent~$i$, we define two payoff functions associated with acting as leader and as follower, denoted by $u_{i,1}(t)$ and  $u_{i,0}(t)$, respectively. Consistent with the literature~\cite{VanSwol01072009,brandt2015unthinking}, the reward associated with acting as a leader increases with the strength of the agent's opinion, either in support or against the specific topic considered. This effect can be captured by the squared 2-norm of the opinion, namely,
\begin{subequations}\label{eq:reward}\begin{equation}\label{eq:reward1}
    u_{i,1}(z)=\rho_i x_i^2,
\end{equation}
where $\rho_i>0$ is a strictly positive coefficient that measures the strength of this mechanism for agent $i$. 

The reward associated with acting as a follower for agent~$i$ is driven by their desire to be socially aligned with others~\cite{klein2023respondents,hillman2023social}. Hence, the corresponding reward function increases the more the agent's opinion deviates from the weighted average opinion of their neighbors. This misalignment can be measured by the squared distance between $i$'s opinion and a weighted average, i.e.,  
\begin{equation}\label{eq:reward0}
    u_{i,0}(z)=\Big(x_i-\sum\nolimits_{j=1}^n W_{ij}x_j\Big)^2.
\end{equation}\end{subequations}
Observe that, while the reward for being a leader only depends on the agent's opinion, the reward for being a follower also depends on the state of the network through the opinions of the neighbors.

We assume that agents revise their leadership according to a replicator equation~\cite{taylor1978replicator}, which is often used in evolutionary game theory to capture the emergent behavior of game-theoretic decision-making processes~\cite{Sandholm2010}. According to this mechanism, the leadership of a generic agent $i$ evolves according to the following ODE:
\begin{equation}\label{eq:leader}
\begin{split}
    \dot y_i=\,&y_i(1-y_i)(u_{i,1}(z)-u_{i,0}(z))\\=\,&\ds y_i(1-y_i) \Big( \rho_i x_i^2 -  \Big( x_i - \sum\nolimits_{j =1}^n W_{ij} x_j \Big)^2 \Big).
    \end{split}
\end{equation}
The leadership dynamics in \eqref{eq:leader} is governed by a complementary tradeoff. On the one hand, agents holding strong opinions tend to increase their leadership, reflecting the emergence of leaders driven by individual conviction. On the other hand, social misalignment reduce leadership incentives. The parameter $\rho_i$ captures the cognitive and social characteristics of agent $i$, regulating the relative importance of the agent's opinion strength and social alignment in shaping their leadership, thus determining whether leadership is primarily driven by conviction or by local consensus. 
Ultimately, the payoff difference $u_{i,1}(z)-u_{i,0}(z)$ captures the balance between individual conviction, promoting leadership through strong opinions, and social legitimacy, penalizing leadership in the presence of disagreement with the neighborhood. Note that the sign of the payoff difference $u_{i,1}(z)-u_{i,0}(z)$ determines the direction of evolution of $y_i$: leadership increases when $u_{i,1}>u_{i,0}$ and decreases when the opposite relation holds true. 

\begin{remark}\label{rem:leader}
    When opinions are fixed, i.e., $x(t)=x^*$, the leadership game defined by the two reward functions in \eqref{eq:reward} is a dominant-strategy game~\cite{Como2021}. In this case, the replicator equation \eqref{eq:leader} converges monotonically to a pure Nash equilibrium, where either $y_i=0$ if \eqref{eq:reward0} is greater than \eqref{eq:reward1}, or $y_i=1$ if the opposite inequality holds true. 
\end{remark}

\subsection{Opinion dynamics}
Together with their degree of leadership, agents revise their opinions over time. Opinions evolve according to a classical opinion dynamics model that accounts for two contrasting mechanisms: the tendency of agents to align their opinions with those of their peers, captured by a linear averaging dynamics~\cite{proskurnikov2017tutorial}, and their attachment to their initial beliefs $x_i(0)\in[-1,1]$. The relative impact of these two mechanisms is traded-off by the agent's leadership $y_i$, whereby agents with a strong leadership are less influenced by others and more attached to their initial belief. To model this behavior, we adopt a continuous-time implementation of the classical Friedkin--Johnsen model~\cite{Friedkin1990}, originally proposed in~\cite{taylor1968towards}, obtaining the following ODE:
\begin{equation}\label{eq:opinion}
    \dot{x}_i = \ds -(1 - y_i) \Big( x_i - \sum\nolimits_{j =1}^n W_{ij} x_j \Big) - y_i(x_i - x_i(0)).
\end{equation}

In \eqref{eq:opinion}, the opinion dynamics is governed by two competing mechanisms. The first term, weighted by $(1-y_i)$, captures social conformity: agents with low degree of leadership tend to reduce their disagreement with the local neighborhood. The second term represents agent anchoring: leaders resist social influence and remain close to their initial opinion.
However, unlike the original formulation of the Friedkin--Johnsen model in \cite{taylor1968towards,Friedkin1990} but similar to~\cite{Alutto2026ifac}, the variable $y_i$ is not a constant, but a time-varying leadership quantity that evolves according to \eqref{eq:leader}.

\begin{remark}\label{rem:opinion}
When the leadership is fixed, i.e., $y(t)=y^*$, \eqref{eq:leader} reduces to a standard continuous-time Friedkin-Johnsen model, which is known to converge to its unique equilibrium~\cite{taylor1968towards,Friedkin1990}, i.e., $x(t)\to(I-W)^{-1}(I-\diag{(y)})x(0)$.
\end{remark}

\subsection{Coupled evolution of opinions and leadership}

\begin{figure}
    \centering
    \includegraphics[width=\linewidth]{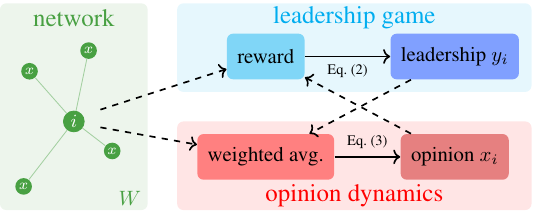}
    \caption{Schematic of the coupled co-evolution of leadership and opinions \eqref{eq:model}.}
    \label{fig:schema}
\end{figure}

In summary, the coupled evolution of opinions and leadership is governed by the following system of $2n$ nonlinear ODEs that couples \eqref{eq:leader} and \eqref{eq:opinion}, for all $i = 1,\dots,n$:
\begin{equation}\label{eq:model}
\begin{cases}
	\dot{x}_i = \ds \! - (1 - y_i)\Big( x_i \!-\! \sum\nolimits_{j =1}^n W_{ij} x_j \!\Big)\! -\! y_i(x_i \!-\! x_i(0)), \\
	\dot{y}_i = \ds y_i(1-y_i) \Big( \rho_i x_i^2 -  \Big( x_i \!- \!\sum\nolimits_{j =1}^n  W_{ij} x_j \Big)^2 \Big),
\end{cases}
\end{equation}
where $\rho_i > 0$ is a scalar parameter and $x_i(0)$ denotes the initial opinion of agent $i$. The model, schematically illustrated in Fig.~\ref{fig:schema}, is parsimonious, in the sense that all cognitive and behavioral features are encapsulated within a unique scalar parameter $\rho_i$ of each agent.

It is worth noticing that, while the two uncoupled subsystems have well-understood behavior and converge to a unique equilibrium, as observed in Remarks~\ref{rem:leader} and~\ref{rem:opinion}, their coupling gives rise to a significantly richer and nontrivial emergent behavior, as we will illustrate in the following. Indeed, since the difference between \eqref{eq:reward1} and \eqref{eq:reward0} ultimately depends on the opinion vector $x(t)$, the evolution of $x(t)$ can ultimately affect the monotonicity of $y(t)$. Conversely, the leadership dynamics feeds back into the opinion dynamics by dynamically modifying the tradeoff between social conformity and individual anchoring, thereby shaping the opinion evolution in a nontrivial way.

\section{General results on the model}\label{sec:general}
The overall leadership-opinion dynamic model in \eqref{eq:model} is well-defined, in the sense that both variables remain in their domain, as shown in the following lemma.

\begin{lemma}\label{lemma:invariance}
    Consider the opinion-leadership model \eqref{eq:model}. If $x(0)\in[-1,1]^n$ and $y(0)\in[0,1]^n$, then the domain $\mathcal D:=[-1,1]^n\times[0,1]^n$ is positively invariant under \eqref{eq:model}.
\end{lemma}
\begin{proof}
    First observe that the right-hand side of \eqref{eq:model} is continuously differentiable with respect to $(x,y)$. Hence, it is Lipschitz-continuous, guaranteeing existence and uniqueness of solutions for all initial conditions in $\mathcal D$ \cite{hale}.
    Note that $\mathcal D$ is compact and convex and at $y_i=0$ and $y_i=1$, $\dot y_i= 0$. While at $x_i=-1$, $\dot x_i\geq y_i(1+x_i(0))\geq 0$, and at $x_i=1$, $\dot x_i\leq -y_i(1-x_i(0))\leq 0$. Thus, at every boundary point of $\mathcal D$, the vector field points inward. Hence, Nagumo's Theorem yields the claim~\cite{blanchini1999set}. 
\end{proof}

Accordingly, we make the following assumption throughout the paper.
\begin{assumption} \label{ass:ass1}
    The initial conditions of \eqref{eq:model} satisfy $(x(0),y(0)) \in \mathcal D$.
\end{assumption}

The next result establishes that the leadership variable remains strictly positive if initially positive.
\begin{lemma}\label{lemma:positivity}
    Consider the opinion-leadership model \eqref{eq:model} under Assumption~\ref{ass:ass1}. If $y_i(0)>0$, then $y_i(t)>0$ for all $t\ge0$.
\end{lemma}
\begin{proof}
Lemma~\ref{lemma:invariance} implies that $x(t)\in[-1,1]$ and $y_i(t)\in[0,1]$. Hence,  $( x_i - \sum_{j =1}^n W_{ij} x_j )^2\leq4$ and so, from \eqref{eq:leader}, $\dot y_i\geq -4y_i$. Finally, using Gronwall's lemma, we bound $y_i(t)\geq y(0)e^{-4t}>0$, which yields the claim.
\end{proof}

We now study the equilibria of the coupled system \eqref{eq:model}. To simplify the notation, for each agent $i$ we define the \emph{local weighted average opinion}
\begin{equation}\label{eq:si}
    s_i(x):=\sum\nolimits_{j=1}^n W_{ij}x_j,
\end{equation}
which represents the opinion perceived by agent~$i$ in its neighborhood.
The next result shows that, at equilibrium, agents can be partitioned into leaders ($y_i^*=1$), complete followers ($y_i^*=0$), and partial followers ($y_i^*\in(0,1)$).
\begin{lemma}\label{lemma:equilibria}
    Consider the opinion-leadership model~\eqref{eq:model} under Assumption~\ref{ass:ass1} and let $(x^*,y^*)\in \mc D$ be an equilibrium point. Define the partition $\mathcal V=\{1,\dots,n\}=\mathcal V_0\cup \mathcal V_1\cup \mathcal V_m$, where 
    $\mathcal V_0:=\{i:\ y_i^*=0\}$, 
    $\mathcal V_1:=\{i:\ y_i^*=1\}$, and 
    $\mathcal V_m:=\{i:\ y_i^*\in(0,1)\}$. 
    Then the following results hold true:
    \begin{itemize}
        \item[(i)] If $i\in\mathcal V_0$, then $x_i^*=s_i(x^*)$.
        \item[(ii)] If $i\in\mathcal V_1$, then $x_i^*=x_i(0)$.
        \item[(iii)] If $i\in\mathcal V_m$ and $\rho_i\neq 1$, then $x_i^* = u_i(x^*)$ or $x_i^*=\ell_i(x^*)$ with 
    \begin{subequations}
        \begin{align}
            u_i(x) &:=\frac{-s_i(x)+\sqrt{\rho_i}\,|s_i(x)|}{\rho_i-1}, \label{eq:ui} \\
            \ell_i(x) &:=\frac{-s_i(x)-\sqrt{\rho_i}\,|s_i(x)|}{\rho_i-1},\label{eq:elli}
        \end{align}
    \end{subequations}
   and 
    \begin{equation}\label{eq:y_star}
    y_i^*=\frac{s_i(x^*)-x_i^*}{s_i(x^*)-x_i(0)},
    \end{equation}
    whenever $s_i(x^*)\neq x_i(0)$, which satisfies $y_i^*\in(0,1)$.
        \item[(iv)] If $i\in\mathcal V_m$ and $\rho_i =1$, $s_i(x^*)=0$ or $2x_i^*=s_i(x^*)$, and $y_i^*$ is given by \eqref{eq:y_star}, when $s_i(x^*)\neq x_i(0)$.
    \end{itemize}
\end{lemma}
\begin{proof}
    At equilibrium, $\dot y_i=0$ implies that for each $i$ either $y_i^*\in\{0,1\}$ or
    \begin{equation}\label{eq:diff-null}
    \rho_i {x_i^*}^2-(x_i^*-s_i(x^*))^2=0,
    \end{equation}
    and $\dot x_i =0$ implies 
    \be \label{eq:dot-xi-eq}(x_i^*-s_i(x^*))(1-y_i^*)=-y_i^*(x_i^*-x_i(0)).\ee
    If $y_i^*=0$, then \eqref{eq:dot-xi-eq} implies $x_i^*=s_i(x^*)$, proving (i).
    If $y_i^*=1$, then \eqref{eq:dot-xi-eq} yields $x_i^*=x_i(0)$, proving (ii).
    Assume now $y_i^*\in(0,1)$. If $\rho_i\neq 1$, \eqref{eq:diff-null} 
    is satisfied if either $x_i^*=u_i(x^*)$ or $x_i^*=\ell_i(x^*)$ defined in \eqref{eq:ui} and \eqref{eq:elli}, respectively.
    Finally, imposing $\dot x_i=0$ gives \eqref{eq:y_star} whenever $s_i(x^*)\neq x_i(0)$. The constraint $y_i^*\in(0,1)$ follows from $i\in\mathcal V_m$. Thus proving (iii).    
    If $\rho_i=1$, \eqref{eq:diff-null} becomes ${x_i^*}^2=(x_i^*-s_i(x^*))^2$, which is equivalent to $s_i(x^*)=0$ or $2x_i^*=s_i(x^*)$, proving (iv).
\end{proof}

We observe that the system \eqref{eq:model} has, in general, multiple equilibria: agent~$i$ can be a leader ($y^*_i=1$), a complete follower ($y^*_i=0$), or a partial follower ($y^*_i\in(0,1)$). 
If agent $i$ is a leader, Lemma~\ref{lemma:equilibria}(ii) establishes that their opinion coincides with their initial opinion, i.e., $x_i^*=x_i(0)$. If agent $i$ is a complete follower, their opinion is the weighted average of their neighbors (Lemma~\ref{lemma:equilibria}(i)), as predicted by classical opinion dynamics models~\cite{proskurnikov2017tutorial}. Interestingly, if agent $i$ is a partial follower, then depending on their parameter $\rho_i$ the agent's opinion can have multiple equilibria, (Lemma~\ref{lemma:equilibria}(iii)-(iv)), showcasing the richness of the model. 
Figure~\ref{fig:ex2} shows a numerical simulation of the coupled opinion-leadership model \eqref{eq:model} on a network with $n=10$ agents to better illustrate its emergent behavior. The trajectories converge to an equilibrium configuration that, according to Lemma \ref{lemma:equilibria}, corresponds to a partition of the agent set into $\mc V_m$ and $\mc V_1$.
In particular, five agents converge to a final leadership level equal to $1$, becoming full leaders and anchoring their opinions close to their initial values. The remaining five agents behave as (partial) followers, adjusting their opinions under the influence of the weighted average of their neighbors. Interestingly, no agent behaves like a complete follower. Indeed, we claim that equilibria where all agents are complete followers are unstable, but a rigorous proof is left for future work, due to space limitations.

The behavior observed in Fig.~\ref{fig:ex2} highlights the nonlinear feedback between opinions and leadership and motivates the research questions: Under which conditions does leadership emerge in the social network? Is it possible to predict a priori which agent will necessarily emerge as a leader and who, instead, will be a follower? Are these features inherent properties of each agent, or do they depend on others?
In the next section, we address these questions by characterizing sufficient conditions under which leadership arises and is sustained at equilibrium.

\begin{figure}
    \centering
    \subfloat[Opinions]{\includegraphics[width=0.5\linewidth]{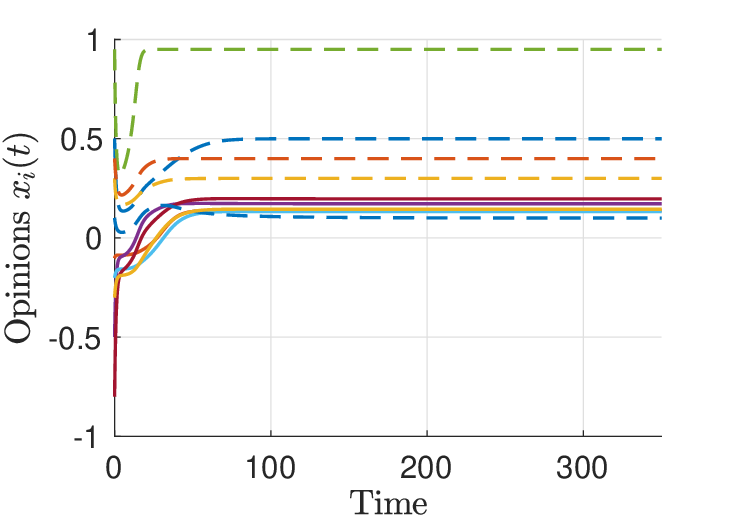}}
    \subfloat[Leadership]{\includegraphics[width=0.5\linewidth]{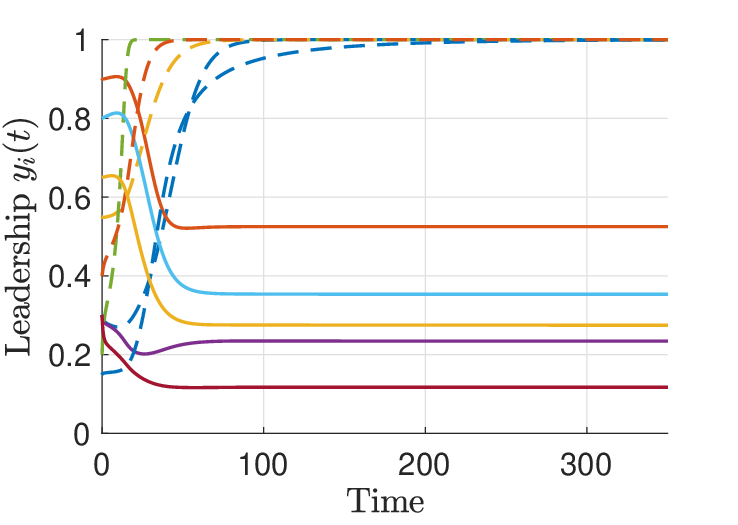}}
    \caption{Numerical simulation of the model \eqref{eq:model} with $n=10$ agents showing the trajectories of (a) opinion and (b) leadership. Dashed trajectories indicate the emerging (five) leaders.}
    \label{fig:ex2}
\end{figure}


\section{Leadership emergence}\label{sec:results}
In this section, we analyze the conditions under which leaders emerge.
In particular, in Section~\ref{sec:leader}, we establish sufficient conditions for the emergence of leaders, and in Section~\ref{sec:follower}, we identify who necessarily acts as a follower.

\subsection{Sufficient conditions for the emergence of leaders}\label{sec:leader}

Recalling the definitions of $u_i(x)$ and $\ell_i(x)$ in \eqref{eq:ui} and \eqref{eq:elli}, for agent $i$ we define the auxiliary functions
\begin{equation}\label{eq:g}G^+_i(x):=x_i-u_i(x),\quad G^-_i(x):=x_i-\ell_i(x),\end{equation}
which measure the signed distance of the opinion $x_i$ from the two critical thresholds $u_i(x)$ and $\ell_i(x)$. In the following, we show that the sign of these auxiliary functions at a certain time determines the derivative of the leadership variable.
In this first work, we restrict our analysis to agents $i$ with $\rho_i>1$, for which tractable sufficient conditions for leadership emergence can be established.
\begin{lemma}\label{lemma:dot-yi}
    Consider the opinion-leadership model \eqref{eq:model} under Assumption~\ref{ass:ass1}.
    For agent $i$ with $\rho_i>1$, if at some time $t\ge 0$, either $G_i^+(x(t))\geq 0$ or $G_i^-(x(t))\leq 0$, then $\dot y_i(t) \geq 0$.
\end{lemma}
\begin{proof}
    By definition, $G_i^+(x(t))\geq 0$ or $G_i^-(x(t))\leq 0$ implies either $x_i(t)\geq u_i(x(t))$ or $x_i(t)\leq \ell_i(x(t))$, respectively. Observe that for $\rho_i>1$ the inequality
    $$ \rho_i x_i^2-(x_i-s_i(x))^2 \geq 0 $$
    is equivalent to $x_i\leq \ell_i(x)$ or $x_i \geq u_i(x)$, since $\ell_i(x)$ and $u_i(x)$ are exactly the two real roots (ordered) of \eqref{eq:diff-null}.
    Hence, from the leadership dynamics \eqref{eq:leader}, it follows that $\dot y_i(t) \ge 0$ for all $(x(t),y(t))$ in $\mc D$.
\end{proof}

Based on Lemma~\ref{lemma:dot-yi}, we define two regions of the state space where the leadership for an agent $i$, with $\rho_i>1$, is non-decreasing:
\begin{equation} \label{eq:omega}
\begin{array}{l}\Omega^+_i:=\{(x,y)\in\mc D: G^+_i(x)>0\},\\
 \Omega^-_i:=\{(x,y)\in\mc D: G^-_i(x)<0\}.\end{array}\end{equation}
Next, we provide sufficient conditions for the positively invariance of these sets, which implies that the leadership is monotonically increasing. 

\begin{proposition}\label{prop:prop1}
    Consider the opinion-leadership model \eqref{eq:model} under Assumption~\ref{ass:ass1}. Consider an agent $i$ with $\rho_i>1$ and $x_i(0) \in (-1,1)$. Then the following hold:
    \begin{itemize}
        \item[(i)] if $(x(0),y(0))\in\Omega^+_i$ and 
        \be\label{eq:Gi-hyp} y_i(0)\! >\! \frac{1}{\rho_i\! -\!1}\max\Big\{\frac{\sqrt{\rho_i}-1}{x_i(0)}, \frac{\sqrt{\rho_i}(1+\sqrt{\rho_i})}{1+x_i(0)}\Big\}, \ee
        then $(x(t),y(t))\in\Omega^+_i$ for all $t \geq 0$.
        \item[(ii)] If $(x(0),y(0))\in\Omega^-_i$ and 
        \be\label{eq:Gi-hyp2} y_i(0)\! >\! \frac{1}{\rho_i\! -\!1}\!\max\Big\{\!\! -\frac{\sqrt{\rho_i}-1}{x_i(0)}, \frac{\sqrt{\rho_i}(1+\sqrt{\rho_i})}{1-x_i(0)}\Big\}, \ee
        then $(x(t),y(t))\in\Omega^-_i$ for all $t \geq 0$.
    \end{itemize}
\end{proposition}
\begin{proof}
    To simplify the analysis, define for $\rho_i >1$ the two scalar quantities,
    \be \label{eq:ai-bi}a_i:=\frac{\sqrt{\rho_i}-1}{\rho_i-1},  \qquad  b_i:=\frac{\sqrt{\rho_i}+1}{\rho_i-1}.\ee
    (i) Initial condition in $\Omega^+_i$ implies $G^+_i(x(0))>0$. We want to prove that the corresponding solution will remain in $\Omega^+_i$ for all $t \geq 0$. 
    Assume by contradiction that the statement is false; then, there exists $ t^*:=\inf\{t>0:\ G^+_i(x(t))\le 0\}$, yielding $G^+_i(x(t))>0$ for all $t<t^*$ and $G^+_i(x(t^*))=0$. By Lemma~\ref{lemma:dot-yi}, $y_i(t)$ is non-decreasing on $[0,t^*)$. We distinguish two cases depending on the sign of $s_i(x(t^*))$.

    \emph{Case 1: $s_i(x(t^*))\ge 0$}. Then, $u_i(x(t^*)) = a_i s_i(x(t^*))$ and $G^+_i(x(t^*))=x_i(t^*)-a_i s_i(x(t^*))$.
    Let us compute the first time derivative of $G^+_i(x(t))$, and evaluate it 
   at time $t=t^*$, using $x_i(t^*)=a_i s_i(x(t^*))$. We obtain:
    \begin{align*}
        \dot G^+_i&(x(t^*)) = (1-y_i) s_i(x) + x_i(0) y_i \\
        &\quad -  a\sum\nolimits_{j=1}^n W_{ij} (1-y_j)s_j(x) -a_i\sum\nolimits_{j=1}^n W_{ij}x_j(0)y_j \\
        &\geq x_i(0) y_i - a\sum\nolimits_{j=1}^n W_{ij} (1-y_j)- a_i\sum\nolimits_{j=1}^n W_{ij}y_j \\
        &\geq x_i(0) y_i(0) - a_i> 0,
    \end{align*}
    where the first inequality follows from the facts that $s_i(x) \geq 0$ for all $i$, $s_j(x) \leq 1$, and $x_j(0)\leq 1$ for all $j$. The second inequality follows from the facts that $y_i(t)$ is non-decreasing for all $t \in [0,\tau)$ from Lemma \ref{lemma:dot-yi} and that $x_i(0)> u_i(0)>0$. The last inequality follows from \eqref{eq:Gi-hyp}, using \eqref{eq:ai-bi}. Hence $\dot G^+_i(x(t^*))>0$.
    However, this is in contrast with the definition of $t^*$, leading to a contradiction.
    Hence Case 1 is impossible.
    
    \emph{Case 2: $s_i(x(t^*))< 0$}. Then $u_i(x(t^*)) = -b_i s_i(x(t^*))$ and $G^+_i(x(t^*))=x_i(t^*)+b_i s_i(x(t^*))$.
    Computing the first derivative of $G^+_i(x(t))$, evaluated at $t=t^*$ and using $x_i(t^*)=- b_i s_i(x(t^*))$, yields
       \begin{align*}
        \dot G^+_i&(x(t^*)) = (1-y_i) s_i(x) + x_i(0) y_i \\
        &\quad +b_i \sum\nolimits_{j=1}^n W_{ij} (1-y_j)s_j(x) + b_i \sum\nolimits_{j=1}^n W_{ij}x_j(0)y_j \\
        &\geq  - 1+(1+ x_i(0)) y_i(0) -b_i > 0, 
    \end{align*}
    which follows from the facts that $s_i(x) \geq -1$ for all $i$, $s_j(x) \leq 1$ and $x_j(0)\leq 1$ for all $j$, $y_i(t)$ is non-decreasing for all $t \in [0,\tau)$, and \eqref{eq:Gi-hyp}, using \eqref{eq:ai-bi}. Hence $\dot G^+_i(x(t^*))>0$, which yields a contradiction also for Case 2, implying that $\Omega_i^+$ is positively invariant.
    
    (ii) The proof for $\Omega^-_i$ is analogous. Define $t^*:=\inf\{t>0:\ G^-_i(x(t))\geq 0\}.$
    Then $G^-_i(x(t))<0$ for all $t<t^*$ and $G^-_i(x(t^*))=0$. We proceed as above and distinguish the sign of $s_i(x(t^*))$.
    
    \emph{Case 1: $s_i(x(t^*))\ge 0$}. Then $\ell_i(x(t^*)) = -b_i s_i(x(t^*))$ and $G^-_i(x(t^*))=x_i(t^*)+b_i s_i(x(t^*))$.
    Computing the first derivative of $G^-_i(x(t))$, evaluating in $t^*$ and using $x_i(t^*)=-b_i s_i(x(t^*))$ yields 
    \begin{align*}
        \dot G&^-_i(x(t^*)) = (1-y_i) s_i(x) + x_i(0) y_i \\
        &\quad +  b_i \sum\nolimits_{j=1}^n W_{ij} (1\! -\! y_j)s_j(x) + b_i \! \sum\nolimits_{j=1}^n W_{ij}x_j(0)y_j \\
        &\leq 1 \!-\! y_i \!+\! x_i(0) y_i \!+ \! b_i\! \sum\nolimits_{j=1}^n\!\! W_{ij} (1-y_j) + \! b_i \! \sum\nolimits_{j=1}^n\!\! W_{ij} y_j  \\
        &\leq 1 - (1 - x_i(0)) y_i(0) +b_i< 0, 
    \end{align*}
    which follows, similar as for i), from \eqref{eq:Gi-hyp2}, using \eqref{eq:ai-bi}. Hence $\dot G^-_i(x(t^*))<0$, which contradicts the necessary condition $\dot G^-_i(x(t^*))\geq 0$. Thus Case 1 is impossible.
    
    \emph{Case 2: $s_i(x(t^*))<0$}. Then $\ell_i(x(t^*)) = a_i s_i(x(t^*))$, hence $G^-_i(x(t^*))=x_i-a_i s_i(x(t^*))$. Computing the first derivative of $G^-_i(x(t))$, evaluating in $t^*$ and using $x_i(t^*)=a_i s_i(x(t^*))$ yields
    \begin{align*}
        \dot G^-_i&(x(t^*)) = (1-y_i) s_i(x) + x_i(0) y_i \\
        &\quad -a_i \sum\nolimits_{j=1}^n\!\! W_{ij} (1-y_j)s_j(x) - a_i \sum\nolimits_{j=1}^n \!\! W_{ij}x_j(0)y_j \\
        &\leq  x_i(0) y_i+a_i \leq x_i(0) y_i(0)+a_i < 0, 
    \end{align*}
    which follows again from the fact that $x_i(0)< l_i(0)<0$ and \eqref{eq:Gi-hyp2}, using \eqref{eq:ai-bi}. This leads to a contradiction, implying that $\Omega_i^-$ is positively invariant.
\end{proof}
\begin{remark}
    Note that conditions \eqref{eq:Gi-hyp}--\eqref{eq:Gi-hyp2} are well-defined. Indeed, for any agent $i$ whose initial condition belongs to either $\Omega_i^+$ or $\Omega_i^-$, it necessarily holds that $x_i(0)\neq 0$. In particular, if $(x(0),y(0))$ is in $\Omega_i^+$ then $x_i(0)>u_i(0)>0$, whereas if it is in $\Omega_i^-$ then $x_i(0)<l_i(0)<0$. Therefore, these conditions ensure that the opinion $x_i$ evolves within a strip that remains separated from the origin, i.e., away from moderate opinions. As will be shown later, this property plays a key role in the emergence of leadership.
\end{remark}

\begin{table}    \caption{Variables and parameters of the model \eqref{eq:model}. }
    \label{tab:parameters}
    \centering
\begin{tabular}{r|l}
$x_i(t)$&opinion of agent $i$ at time $t$\\
$y_i(t)$&leadership of agent $i$ at time $t$\\
$\rho_i$&behavioral parameter that characterizes agent $i$\\
$s_i(x)$&local weighted average opinion for $i$, see \eqref{eq:si} \\
$G_i^{\pm}(x)$&signed distances with respect to thresholds for $i$, see \eqref{eq:g}\\
    \end{tabular}
\end{table} 

Building on the technical result presented above, we can now state our main result, which establishes sufficient conditions for the emergence of leaders. Table~\ref{tab:parameters} lists all the variables and parameters that appear in the main statements.

\begin{theorem}\label{theo:theo-leaders}
    Consider the opinion-leadership model \eqref{eq:model} under Assumption~\ref{ass:ass1}. Consider an agent $i$ with $\rho_i>1$ and $x_i(0) \in (-1,1)$. If either
    \begin{itemize}
        \item[(i)] Inequality \eqref{eq:Gi-hyp} holds and $(x(0),y(0))\in\Omega^+_i$, or
        \item[(ii)] Inequality \eqref{eq:Gi-hyp2} holds and $(x(0),y(0))\in\Omega^-_i$, 
    \end{itemize}
    with $\Omega_i^{\pm}$ defined in \eqref{eq:omega}, then 
    $\lim_{t\to \infty} y_i(t) = 1.$
\end{theorem}
\begin{proof}
    The proof is divided into two steps. First, we show that $y_i(t)$ admits a finite limit as $t\to\infty$. Then, we argue by contradiction: assuming that such a limit satisfies $y_i^*<1$, we show that this cannot occur if either (i) or (ii) holds.

    By Proposition \ref{prop:prop1}, for all $t\ge 0$, either $G_i^+(x(t))>0$ or $G_i^-(x(t))<0$, depending on the initial condition. By Lemma \ref{lemma:dot-yi}, this implies that $\dot y_i(t) \geq 0$ for all $t\ge 0$. 
    Moreover, by Lemma~\ref{lemma:invariance}, $y_i(t)\in[0,1]$ for all $t\ge 0$. Hence $y_i(t)$ is bounded and non-decreasing. By the monotone convergence theorem \cite{rudin1976principles}, the limit $\ds \lim_{t \to \infty} y_i(t) =y_i^*\in [y_i(0), 1]$ exists. 
    
    Assume by contradiction that $y_i^* < 1$. 
    Since $y_i(0)>0$ from \eqref{eq:Gi-hyp}--\eqref{eq:Gi-hyp2}, and $y_i$ is non-decreasing, $y_i^*\in(0,1)$ and hence $y_i^*(1-y_i^*)>0$. Given that $y_i(t)$ is monotone and convergent and $\dot y_i$ is continuous, it follows that $\dot y_i(t)\to 0$ as $t\to\infty$. In particular, for any sequence $t_k\to\infty$, $\dot y_i(t_k)\to 0$.
    
    Since $x_i(t)\in[-1,1]$ for all $t\ge 0$ by Lemma~\ref{lemma:invariance} and $x_i$ is continuously differentiable, there exists a sequence $t_k\to\infty$ such that $\dot x_i(t_k)\to 0$. Moreover, since $x(t)\in[-1,1]^n$ for all $t\ge 0$ and $[-1,1]^n$ is compact, we can extract a subsequence such that $x(t_k)\to x^* \in [-1,1]^n$ and $s_i(x(t_k))\to s_i(x^*)$ for agent $i$ satisfying the assumptions. 
    Define $h_i(x):= \rho_i {x_i}^2-(x_i-s_i(x))^2.$
    Therefore, passing to the limit along $t_k$ in \eqref{eq:leader} yields
    \begin{align*}
        0=\lim_{k\to\infty}\dot y_i(t_k)
        &=\lim_{k\to\infty} y_i(t_k)\bigl(1-y_i(t_k)\bigr)\,h_i(x(t_k))\\
        &= y_i^*(1-y_i^*)\lim_{k\to\infty} h_i(x(t_k)),
    \end{align*}
    and $y_i^*(1-y_i^*)>0$ implies $\ds \lim_{k\to\infty} h_i(x(t_k))=0$.
    By continuity with respect to $(x_i,s_i)$,
    \begin{equation}\label{eq:h-limit}
        \rho_i {x_i^*}^{2}-\bigl(x_i^*-s_i(x^*)\bigr)^2=0.
    \end{equation}
    This implies that $x_i^*=u_i(x^*)$ or $x_i^*=\ell_i(x^*)$. In particular, since $|s_i(x^*)|\le 1$, then $|x_i^*|\le b_i$, defined in \eqref{eq:ai-bi}.
    Moreover, since $\dot x_i(t_k)\to 0$, considering \eqref{eq:opinion} and passing to the limit along $t_k$ yields
    \begin{equation}\label{eq:xdot-limit}
        0= -x_i^*+(1-y_i^*)s_i(x^*)+y_i^*x_i(0).
    \end{equation}
    We now treat the two cases (i) and (ii) separately.

    In Case (i), by Proposition~\ref{prop:prop1}, $G_i^+(x(t_k))>0$ and therefore $G_i^+(x^*)\ge 0$. Combining this with \eqref{eq:h-limit} implies necessarily that $x_i^*=u_i(x^*)$ (the alternative root would violate $G_i^+(x^*)\ge 0$).
    Using \eqref{eq:xdot-limit}, $s_i(x^*)\ge -1$, $|x_i^*|\le b_i$, and $y_i^* \in [y_i(0),1)$, we obtain the following contradiction, i.e.,  Case (i) is impossible:
    \begin{align*}
        0 &= -x_i^*+(1-y_i^*)s_i(x^*)+y_i^*x_i(0)\\
        &\ge -x_i^* -(1-y_i^*) + y_i^*x_i(0)\\
        &= -x_i^* -1 + y_i^*(1+x_i(0))\\
        &\ge -b_i -1 + y_i(0)(1+x_i(0))>0,
    \end{align*}
    where the last inequality is ensured from \eqref{eq:Gi-hyp}.

    In Case (ii), by Proposition~\ref{prop:prop1}, $G_i^-(x(t_k))<0$ and $G_i^-(x^*)\le 0$. Together with \eqref{eq:h-limit}, this implies $x_i^*=\ell_i(x^*)$. From \eqref{eq:xdot-limit}, using $s_i(x^*)\le 1$, $|x_i^*|\le b_i$, and $y_i^*\in [y_i(0),1)$, we obtain the following contradiction:
    \begin{align*}
        0 &= -x_i^*+(1-y_i^*)s_i(x^*)+y_i^*x_i(0)\\
        &\le -x_i^* +1 - y_i^*(1-x_i(0))\\
        &\le b_i +1 - y_i(0)(1-x_i(0))<0,
    \end{align*}
    where the last inequality follows from \eqref{eq:Gi-hyp2}. 
    
    Hence, in both cases the assumption $y_i^*<1$ leads to a contradiction. 
    Therefore, we conclude that $y_i^*=1$.
\end{proof}

Theorem~\ref{theo:theo-foll} establishes sufficient conditions for emergence of a leader. Interestingly, a key condition is that  $\rho_i>1$. The parameter $\rho_i$ regulates the relative weight between individual conviction and social misalignment in the leadership dynamics.  When $\rho_i>1$, strong opinions may compensate for the cost of being misaligned with the neighborhood, thereby enabling the emergence of leaders under appropriate conditions. When $\rho_i<1$, social alignment always plays a stronger role for the agent, and leadership tends to decrease unless the agent is already well aligned with their neighborhood, typically promoting more consensus-oriented outcomes. 
The impact of $\rho_i$ on the conditions provided by Theorem~\ref{theo:theo-leaders} can be observed in Fig.~\ref{fig:rho}.

\begin{figure}
    \centering
    \includegraphics[width=0.7\linewidth]{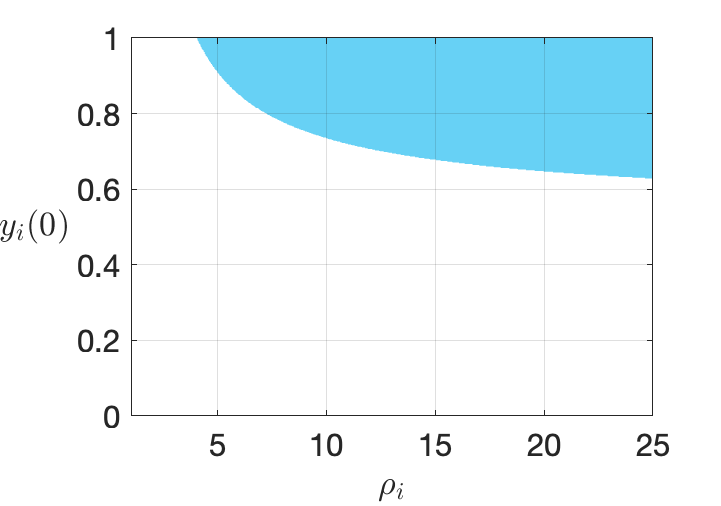}
    \caption{Admissible region in the $(\rho_i, y_i(0))$ plane for a fixed agent $i$. The shaded area represents the set of values of the pair $(\rho_i, y_i(0))$ for which, with $x_i(0)$ fixed, the conditions of Theorem~\ref{theo:theo-leaders} are satisfied.}
    \label{fig:rho}
\end{figure}

\begin{remark}
Conditions \eqref{eq:Gi-hyp}--\eqref{eq:Gi-hyp2} depend only on the agent's initial opinion and leadership level, their characteristics captured by parameter $\rho_i$, and neighbors' initial opinions, but are independent of the subsequent evolution of the other agents. Hence, Theorem~\ref{theo:theo-leaders} characterizes agents whose initial resistance to social influence make them inevitably emerge as leaders in the social group.
\end{remark}
\begin{remark}\label{rem:L}
    The conditions provided by Theorem~\ref{theo:theo-leaders} are sufficient but not necessary. Indeed, agents who do not satisfy them may still become leaders due to the collective evolution of opinions. Given $\mc V_1$ characterized in Lemma~\ref{lemma:equilibria}(ii), let $\mathcal L \subseteq \mathcal V_1$ denote the subset of agents satisfying the assumptions of Theorem~\ref{theo:theo-leaders}. Then, every agent in $\mathcal L$ is guaranteed to asymptotically emerge as a leader.
\end{remark}
Motivated by these remarks, we now complement the previous result by deriving conditions that guarantee the emergence of followers.

\subsection{Identifying agents who necessarily act as followers}\label{sec:follower}

In this section, we derive conditions to establish who necessarily emerges as a follower. 
To this end, we exploit the set $\mathcal L$ introduced in Remark~\ref{rem:L}, consisting of agents that are guaranteed to asymptotically emerge as leaders, i.e., that satisfy condition (i) or (ii) in Theorem \ref{theo:theo-foll}. For every $i\in\mathcal L$, one has $y_i(t)\to 1$ and $x_i(t)\to x_i(0)$. Hence, for every $\varepsilon>0$, there exists $T_\varepsilon\geq 0$ such that
\begin{equation}\label{eq:leaders-eps}
|x_i(t)-x_i(0)|\le \varepsilon,
\qquad \forall i\in\mathcal L,\ \forall t\ge T_\varepsilon.
\end{equation}

The following result gives sufficient conditions ensuring that an agent in $\mathcal F:=\mc V\setminus\mc L$ does not asymptotically emerge as a leader.
First, let us define 
\be \label{eq:etaj} \eta_j:=\sum\nolimits_{i\in\mathcal L}W_{ji}>0, \quad \forall j \in \mc F.\ee

\begin{theorem}\label{theo:theo-foll}
    Consider the opinion-leadership model \eqref{eq:model} under Assumption~\ref{ass:ass1}. Suppose there is a nonempty set $\mc L$ of agents satisfying the assumptions of Theorem~\ref{theo:theo-leaders}.
    Let $j\in\mathcal F$ and assume that there exist $\varepsilon>0$ and $m_j\in(0,1]$ such that the following three conditions are jointly satisfied:
    \begin{subequations}
    \begin{align}
        &-m_j <  \,x_j(0)<0, \label{eq:xj0}\\
       & 1\!-\!(1-\varepsilon)\eta_j -\! m_j < \! \sum_{i\in\mathcal L} W_{ji}x_i(0) \!<\! (1-\varepsilon)\eta_j\!-\!1, \label{eq:hyp2} \\
       & \rho_j m_j^2< \!\Big(\!- 1+(1-\varepsilon) \eta_j- \!\sum_{i\in\mathcal L} W_{ji}x_i(0) \! -m_j\Big)^2. \label{eq:mj-cond2}
    \end{align}
    \end{subequations}
   Then, $\lim_{t\to \infty} y_j(t) \in [0,1).$
\end{theorem}
\begin{proof}
    The proof consists of three steps. First, we derive eventual bounds on the local weighted average opinion $s_j(x(t))$ by exploiting the asymptotic behavior of the agents in $\mathcal L$. Second, we show that the interval $(-m_j,0)$ is eventually reached and then positively invariant for $x_j$. Third, we prove that, once $x_j(t)$ remains in such an interval, the quantity $h_j(x(t))$ is strictly negative, which implies that $y_j(t)$ cannot converge to $1$.

    Fix $\varepsilon>0$, and let $T_\varepsilon\ge 0$ be such that \eqref{eq:leaders-eps} holds. First, for any $j\in\mathcal F$ and for all $t\geq0$, decompose
    \be \label{eq:sj-1}s_j(x(t))=\sum\nolimits_{k\in\mathcal F}W_{jk}x_k(t)+\sum\nolimits_{i\in\mathcal L}W_{ji}x_i(t).\ee
    Regarding the influence by agents in $\mc F$, by Lemma \ref{lemma:invariance}, $x(t)\in[-1,1]^n$; then, using \eqref{eq:etaj} we obtain 
    \be \label{eq:sj-2}-1\!+\! \eta_j \!=\! -\!\! \sum_{k\in\mathcal F}\! W_{jk}\leq \! \sum_{k\in\mathcal F}\! W_{jk}x_k(t)\le \! \sum_{k\in\mathcal F}\! W_{jk}\!=\! 1\!-\!\eta_j,\ee
    for all $t \geq 0$, since $W$ is row-stochastic. 
    Moreover, regarding the influence by agents in $\mc L$, by \eqref{eq:leaders-eps}, for all $t\geq T_\varepsilon$,
    \begin{subequations}
        \be \label{eq:sj-3a} \sum_{i\in\mathcal L}W_{ji}x_i(t)\le \sum_{i\in\mathcal L}W_{ji} \bigl(x_i(0)+\varepsilon\bigr) .\ee
        \be \label{eq:sj-3b} \sum_{i\in\mathcal L}W_{ji}x_i(t)\ge \sum_{i\in\mathcal L}W_{ji} \bigl(x_i(0)-\varepsilon\bigr) .\ee        
    \end{subequations}
    Combining \eqref{eq:sj-2} and \eqref{eq:sj-3a} in \eqref{eq:sj-1} yields for all $t\ge T_\varepsilon$,
    \begin{subequations}
    \begin{align}
        s_j(x(t)) \le \bar s_j := 1-(1-\varepsilon)\eta_j+\sum_{i\in\mathcal L}W_{ji}x_i(0), \label{eq:upper-bounds} \\
        s_j(x(t)) \ge \underline s_j := -1 + (1-\varepsilon) \eta_j + \sum_{i\in\mathcal L}W_{ji}x_i(0). \label{eq:lower-bounds}
    \end{align}
    \end{subequations}
    
    We now show that, once $x_j(t)$ enters the interval $(-m_j,0)$, it cannot leave it.
    Assume that $x_j(T_\varepsilon)\in(-m_j,0)$ and suppose by contradiction that there exists $t^*:=\inf\{t\ge T_\varepsilon:\ x_j(t)\notin(-m_j,0)\}$. By continuity, $x_j(t^*)\in\{-m_j,0\}$ and $x_j(t)\in(-m_j,0)$ for $t\in[T_\varepsilon,t^*)$.
    If $x_j(t^*)=0$, then necessarily $\dot x_j(t^*)\ge0$. However, from \eqref{eq:opinion},
    $$\dot x_j(t^*)=(1-y_j(t^*)) s_j(x(t^*)) + y_j(t^*) x_j(0).$$
    Using \eqref{eq:xj0}, \eqref{eq:hyp2}, \eqref{eq:upper-bounds}, and $y_j(t^*)\in[0,1]$, we obtain $\dot x_j(t^*)<0$, which is a contradiction.
    If $x_j(t^*)=-m_j$, then necessarily $\dot x_j(t^*)\le0$.
    Again from \eqref{eq:opinion},
    $$\dot x_j(t^*)= m_j+(1-y_j(t^*))\,s_j(x(t^*)) + y_j(t^*)\,x_j(0).$$
    By \eqref{eq:lower-bounds} and $y_j(t^*)\in[0,1]$,
    $$(1-y_j(t^*))s_j(x(t^*)) + y_j(t^*)x_j(0)\ge \min\{\underline s_j,x_j(0)\}.$$
    Then \eqref{eq:xj0} and \eqref{eq:hyp2} imply $\min\{\underline s_j,x_j(0)\}>-m_j$, hence $\dot x_j(t^*)>0$, again a contradiction.
    Therefore, if $x_j(T_\varepsilon)\in(-m_j,0)$, then $x_j(t)\in(-m_j,0)$, for all $t\geq T_\varepsilon$.

    We next show that there exists a finite time $T\ge T_\varepsilon$ such that $x_j(T)\in(-m_j,0)$.
    Assume by contradiction that $x_j(t)\notin(-m_j,0)$ for all $t\ge T_\varepsilon$.
    If $x_j(t)\le -m_j$ for all $t\ge T_\varepsilon$, then
    \begin{align*}
        \dot x_j(t) &= -x_j + (1-y_j(t))s_j(x(t))+y_j(t)x_j(0) \\
        &\ge m_j + \min \{\underline{s}_j,x_j(0)\} >0,
    \end{align*}
    which follows from \eqref{eq:xj0}, \eqref{eq:hyp2}, \eqref{eq:lower-bounds}, and $y_j(t)\in[0,1]$. Hence, the vector field points strictly to the right whenever $x_j(t)\le -m_j$, which rules out the possibility that the trajectory remains there for all $t\geq T_\varepsilon$.
    If $x_j(t)\ge 0$ for all $t\ge T_\varepsilon$, then
    \begin{align*}
        \dot x_j(t) &= -x_j + (1-y_j(t))s_j(x(t))+y_j(t)x_j(0) <0,
    \end{align*}
    which follows from \eqref{eq:xj0}, \eqref{eq:hyp2}, \eqref{eq:upper-bounds}, and $y_j(t)\in[0,1]$. Thus, the vector field points strictly to the left whenever $x_j(t)\geq 0$, which again yields a contradiction. Hence, there exists a time $T\ge T_\varepsilon$ such that $x_j(T)\in(-m_j,0)$. By the previous argument, this implies $x_j(t)\in(-m_j,0)$ for all $t\ge T $.

    For all $t\geq T$, using the reverse triangle inequality together with $x_j(t)\in(-m_j,0)$ and $s_j(x(t))\leq \overline s_j<0$, we obtain
    \begin{align*}
    |x_j(t)-s_j(x(t))| \ge \bigl||s_j(x(t))|-|x_j(t)|\bigr| \ge -\bar s_j-m_j.
    \end{align*}
    Therefore,
    \begin{align*}
        h_j(x(t))&=\rho_j x_j(t)^2-\bigl(x_j(t)-s_j(x(t))\bigr)^2\\
        &\le \rho_j m_j^2-(-\bar s_j-m_j)^2<0,
    \end{align*}
    where the last inequality follows from \eqref{eq:mj-cond2}. Hence $\dot y_j(t)=y_j(t)(1-y_j(t))h_j(x(t))\le 0$ for all $t\ge T$, thus $y_j$ is non-increasing on $[T,\infty)$. Since $y_j(t)\in[0,1]$ for all $t$ (Lemma \ref{lemma:invariance}), it converges to a limit $y_j^*:= \lim_{t\to \infty} y_j(t) \in[0,1]$.

    Finally, since $y_j(0)<1$ and $y_j=1$ is an equilibrium of \eqref{eq:leader}, uniqueness of solutions implies that $y_j(t)<1$ for all $t\geq 0$. Hence, $y_j^*<1$, which concludes the proof.
\end{proof}

\begin{remark}
    In Theorem~\ref{theo:theo-foll}, the role of the set $\mathcal L$ is crucial, as the agents in $\mathcal L$ generate an asymptotically persistent contribution to the local weighted average opinion $s_j(x(t))$. Conditions \eqref{eq:xj0}--\eqref{eq:hyp2} guarantee that, for sufficiently large times, the opinion of agent $j$ is driven into the interval $(-m_j,0)$ and remains confined there. Condition \eqref{eq:mj-cond2} then ensures that, within such an interval, the misalignment term dominates the conviction term, so that acting as a leader is not advantageous.
\end{remark}

Figure~\ref{fig:ex3} reports a simple simulation on a network with two nodes, to illustrate the two behaviors predicted by Theorems~\ref{theo:theo-leaders} and~\ref{theo:theo-foll}. Specifically, the agent whose trajectories are in blue satisfies the conditions in Theorem~\ref{theo:theo-leaders} and, consistently, emerges as a leader
 with $y_i(t)\to 1$. On the contrary, the other agent (whose trajectories are in orange) satisfies the conditions in Theorem~\ref{theo:theo-leaders} and remains a follower, as predicted by our theoretical findings.

\begin{figure}
    \centering
    \subfloat[Opinions]{\includegraphics[width=0.5\linewidth]{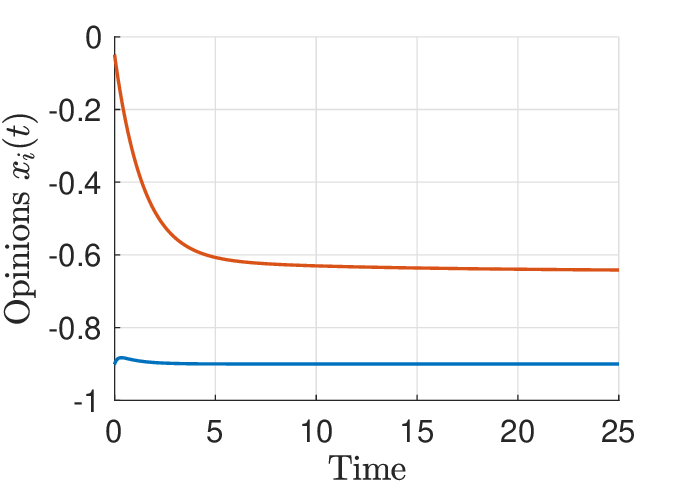}}
    \subfloat[Leadership]{\includegraphics[width=0.5\linewidth]{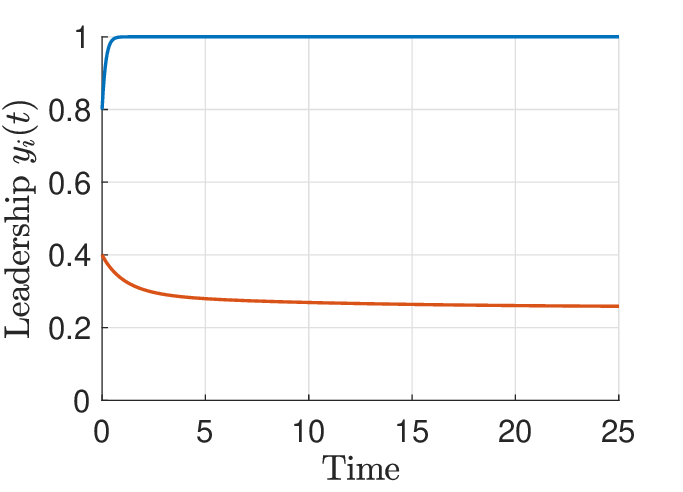}}
    \caption{Numerical simulation of the model \eqref{eq:model} with $n=2$ agents showing the trajectories of (a) opinion and (b) leadership. One agent (blue) satisfies the assumptions of Theorem~\ref{theo:theo-leaders} and becomes a leader, whereas the other (red) satisfies the assumptions of Theorem~\ref{theo:theo-foll} and remains a follower.}
    \label{fig:ex3}
\end{figure}

\section{Conclusion}\label{sec:conclusion}
We proposed a dynamical model for the co-evolution of opinions and leadership in social networks, where leadership is treated as an endogenous state variable rather than as a fixed exogenous trait. The main contribution of the paper is the derivation of sufficient conditions for the emergence of leaders and followers, shedding light onto nontrivial social phenomena. 
Future research directions include integrating the model with empirical data for validation and calibration, as well as designing intervention strategies to shape leadership emergence through opinions or network structure and promote fairness across the social network.

\bibliography{leader_bib}
\bibliographystyle{IEEEtran}
\end{document}